%% file: arxiv-partial_ignorance_v2.tex
\documentclass[12pt,reqno]{amsart}
\usepackage{amsthm}
\usepackage{amsmath}
\usepackage{latexsym}
\usepackage{amsfonts}
\usepackage{amssymb}
\usepackage{color}
\usepackage{bbm,dsfont}
\usepackage{graphicx}
\usepackage{hyperref}
\usepackage{enumerate}
\usepackage{comment}
\usepackage{caption}
\usepackage{subcaption}
\usepackage{amsaddr}

\newtheorem{proposition}{Proposition}
\newtheorem{theorem}{Theorem}

\theoremstyle{definition}

\newtheorem{example}{Example}
\newtheorem{definition}{Definition}

\newcommand{\R}{\mathbb{R}} %real
\newcommand{\C}{\mathbb{C}} %complex
\newcommand{\half}{\tfrac{1}{2}} %half
\newcommand{\hi}{\mathcal{H}} %Hilbert space
\newcommand{\lh}{\mathcal{L(H)}} %bounded linear operators
\newcommand{\lhs}{\mathcal{L}_s(\hi)} %bounded seladjoint linear operators
\newcommand{\no}[1]{\left\|#1\right\|} %norm
\newcommand{\tr}[1]{{\rm tr}\left[#1\right]} %trace
\newcommand{\ran}{\textrm{ran}} %range
\newcommand{\id}{\mathbbm{1}} %identity operator
\renewcommand{\rho}{\varrho}

\newcommand{\rank}[1]{\mathrm{rank}(#1)} %rank

\newcommand{\vr}{\mathbf{r}} %u
\newcommand{\vsigma}{\boldsymbol{\sigma}} %sigma
\newcommand{\M}[2]{\mathcal{M}_{#1,#2}} 
\newcommand{\Mr}[2]{\mathcal{M}^{row}_{#1,#2}} 
\newcommand{\Mall}{\mathcal{M}} 
\newcommand{\uleq}{\preceq} %ultra weak majorization ?
\newcommand{\nuleq}{\npreceq}

\newcommand{\Mo}{\mathsf{M}}%generic observable
\newcommand{\No}{\mathsf{N}}%generic observable
\newcommand{\psuc}{P_{\mathrm{succ}}} %success probability

\begin{document}\setlength{\arraycolsep}{2pt}

\title[]{Communication of partial ignorance with qubits}

\author{Teiko Heinosaari and Oskari Kerppo}
\address{QTF Centre of Excellence, Turku Centre for Quantum Physics, Department of Physics and Astronomy, University of Turku, FI-20014 Turku, Finland}

\begin{abstract}
We introduce a class of communication tests where the task is to communicate partial ignorance by means of a physical system.
We present a full characterization of the implementations of these tests in the qubit case and partial results for qudits.
A peculiar observation is that two physical systems with the same operational dimensions may differ with respect to implementations of these tasks, as is shown to be the case for the qubit and rebit. Finally, we consider the natural question whether some of the communication tests are more difficult than others. A new preordering that we call the ultraweak matrix majorization is presented to answer this question in a theory-independent way.
\end{abstract}

\maketitle

%%%%%%%%%%%%%
\section{Introduction}
%%%%%%%%%%%%%

In any ordinary communication scenario one person (Alice) sends a physical object to another person (Bob). 
By observing, or measuring, the object Bob hopes to recover the message that Alice has encoded into the system. 
The possible messages correspond to different states of the system and thereby the measurement outcome obtained by Bob should depend on the encoded message.  
In the optimal case each state leads to a different measurement outcome. 
If that kind of measurement exists for a set of states, then the states are called distinguishable. 

The basic limitation of the qubit as a communication medium is that, although the qubit has infinitely many pure states, no more than two can be distinguishable.
This means that Alice can communicate to Bob a message reliably only if there are just two alternatives, for instance 'yes' and 'no'.
If Alice and Bob are instead using a communication with $n$ messages, then the total error probability is at least $1-2/n$. 
More generally, a $d$-dimensional quantum system has exactly $d$ distinguishable states and the error with $n>d$ messages is at least $1-d/n$. 
This simple but fundamental result is sometimes called the basic decoding theorem \cite{QPSI10}.
The maximal number of distinguishable states of a quantum system links directly to the respective Hilbert space dimension.
In fact, the maximal number can be taken as the definition of the operational dimension of any physical system \cite{Hardy03}.
This dimension makes sense also for hypothetical systems, such as the real Hilbert space qubit. 
The operational dimension of the real Hilbert space qubit is two, the same as the ordinary qubit.

In the present work we discuss a variant of the usual communication task. 
In this scenario, Alice is trying to transmit to Bob information about which of the alternatives Bob should not choose. We call this information partial ignorance and we will formulate the communication task as a specific kind of test. 
The success probability in the test then tells if the communication of partial ignorance has been successful or not; after a successful communication Bob should be able to act as efficiently as Alice.
Remarkably, the optimal communication of partial ignorance requires totally different kind of encoding and decoding setup than the usual communication. 

The possible ways to transmit partial ignorance with a quantum system are again linked to the respective Hilbert space dimension. 
We provide a full characterization of the solvable tests in the qubit case and derive certain solutions in the general qudit case. 
Interestingly, the so-called symmetric informationally complete (SIC) measurements \cite{Zauner99,ReBlScCa04} are useful in the communication task of partial ignorance. 
We will also show that with the real Hilbert space qubit one cannot communicate as much partial ignorance as with the ordinary qubit, even if the operational dimension is two for both of them. 
We present a systematic way of listing the communication tests that a physical system, existing or hypothetical, can be used to transmit.
This introduced communication table can be seen as refined notion of the operational dimension, the latter determining the diagonal of a communication table.
Finally, we introduce a preorder on the set of matrices, which we call ultraweak matrix majorization as it is weaker than the matrix majorization \cite{Dahl99} and weak matrix majorization \cite{PeMaSi05}. 
The ultraweak matrix majorization gives a theory independent classification of the communication tasks of partial ignorance. 

%%%%%%%%%%%%%%%%%%%%%%
\section{Communication test of partial ignorance}\label{sec:ignorance}
%%%%%%%%%%%%%%%%%%%%%%

We start by formulating the usual communication scenario as a test.
In this test there is an array of $n$ boxes, $n\geq 2$.
Charlie, organizing a communication test to Alice and Bob, picks randomly one box and hides a candy into it.
Alice and Bob are separated and Charlie tells the correct box to Alice but not to Bob.
Alice is allowed to send one qubit to Bob, who then has to open one box based on the information obtained by measuring the qubit. 
Alice and Bob can agree on the encoding and decoding of information before the test starts, but cannot alter it during the test.
Encoding corresponds to some fixed set of states, while decoding corresponds to some fixed measurement apparatus, which we call a measurement.
Alice and Bob pass the test if Bob makes no errors when the test is run repeatedly. 
In order to be successful, Alice must be able to communicate the index of the correct box to Bob with the transmitted qubit.
A qubit system has no more than two perfectly distinguishable states. 
For this reason, Alice can communicate the correct box to Bob without any error only when $n=2$.

We are then considering a variation of the previous test. 
The setting is otherwise the same, but now Charlie tells Alice one box where he will \emph{not} put the candy. 
Again, Alice is allowed to send one qubit to Bob and they can agree on the encoding and decoding of messages before the test starts. 
The aim of Alice and Bob is to maximize the probability of Bob opening the box which contains the candy.
The test is run several times so that Charlie can estimate the success probability of Alice and Bob.

Since even Alice does not know the location of the candy, we obviously cannot expect Bob always to open the box with the candy inside. 
For this reason, we must compare their success probability to the success probability in the case when Alice herself is trying to choose the correct box.
The best that Alice can do is to pick randomly one of the $n-1$ boxes, leaving aside the single box that she knows to be empty from what Charlie told to her.
Therefore, the success probability of Alice is $\tfrac{1}{n-1}$, and to this number we will compare the success probability of Alice and Bob when their communication is limited in some way, e.g. to one qubit passing from Alice to Bob.
A possible deviation in the maximal success probabilities in these two scenarios reveals a limitation of the communication medium in question.

There is one important detail in the game that needs to emphasized. 
When Charlie tells the location of one empty box to Alice, the rules of the test don't demand him to have already chosen the location of the candy; he simply must single out one box where he will not put the candy.
Charlie is even allowed to observe what Alice does before he puts the candy in one of the boxes and gives it to Bob, and he can use this information to minimize the success probability of Alice and Bob.

Before we can study whether Alice and Bob can succeed in the test we need to understand what Alice and Bob need to aim for.
We denote by $p(j|i)$ the probability that Bob chooses the box $j$ when Charlie has revealed the box $i$ to be empty.
These conditional probabilities are the variables that Alice and Bob try to control by adjusting the encoding and decoding of messages.
We write these probabilities in a matrix form, $C_{ij} = p(j|i)$, and call any such matrix as a \emph{communication matrix}.
Communication matrices hence correspond to row-stochastic matrices, i.e., matrices with non-negative entries
with each row summing to 1.
If Bob never opens a box that Charlie told to be empty, then the diagonal of the  matrix $C$ is zero, $C_{11}=C_{22}=\cdots= C_{nn}=0$.
We do not, however, set this as a requirement for allowed communication matrices as Alice and Bob only aim to maximize their success probability and other details are irrelevant.  
As we will shortly conclude, the optimal communication matrix is determined to have zeros in the diagonal, but this will be a consequence rather than assumption.
Summarizing the previous discussion, we conclude that the relevant strategy and actions of Alice and Bob are fully described by a communication matrix $C$.

For a communication matrix $C$, we denote by $\psuc(C)$ the success probability that $C$ leads to when Charlie tries to prevent Alice and Bob to find the candy, within the rules of the test.
Charlie cannot put the candy to the box that he said to be empty, but he can put it to any other box. 
Charlie can also freely choose which box he declares to be empty.
This can be beneficial if he learns, for instance, that Alice and Bob are using a preplanned strategy that makes them to always fail when Charlie tells the first box to be empty.
As Charlie can control the locations of the chosen empty box and the candy, we have
\begin{align}
\psuc(C) = \min \{ C_{ij} : i \neq j \} \, .
\end{align}
The unique optimal communication matrix, denoted as $C^{\mathrm{opt}}_n$, is therefore the $n\times n$ matrix
\begin{equation}\label{eq:optimal-n}
C^{\mathrm{opt}}_n = \frac{1}{n-1} \left[\begin{array}{ccccc} 0 & 1 & 1 & \cdots & 1\\1 & 0 & 1 & \cdots & 1 \\ 1 & 1  & 0 & & 1 \\ \vdots & &   & \ddots  \\ 1 & \cdots & \cdots & 1 & 0 \end{array}\right] \, .
\end{equation}
We remark that the optimal communication matrix is the same if Charlie chooses the index of the empty box with the uniform probability but is still otherwise freely controlling the location of the candy.
In that case, the success probability is given as
\begin{align}
\psuc'(C) = \tfrac{1}{n} \sum_{i=1}^n \min_j \{ C_{ij} : i \neq j \} \, .
\end{align}
This quantification leads to different numerical values than $\psuc$, but the unique optimal communication matrix is still $C^{\mathrm{opt}}_n$.
We conclude that \emph{Alice and Bob can pass the communication test if and only if they can implement the communication matrix $C^{\mathrm{opt}}_n$}.

%%%%%%%%%%%%%%%%%%%%%%
\section{Qudit implementation and uniformly antidistinguishable states}\label{sec:qudit}
%%%%%%%%%%%%%%%%%%%%%%

A physical implementation of a communication matrix $C$ means that Alice is using some states to encode messages and Bob is then performing a measurement on the system that Alice sends to him.
We say that a communication matrix $C$ has a qudit implementation if there are $n$ qudit states $\varrho_1,\ldots,\varrho_n$ and a measurement $\Mo$ (i.e. POVM) with outcomes $1,\ldots,n$ such that
\begin{align}
C_{ij} = \tr{ \varrho_i \Mo(j) } \, .
\end{align}
In particular, the optimal communication matrix $C^{\mathrm{opt}}_n$ has a qudit implementation if there are qudit states $\varrho_1,\ldots,\varrho_n$ and a measurement $\Mo$ such that
\begin{align}\label{eq:anti-uni}
\forall i \neq j: \quad \tr{\varrho_i \Mo(j)} = \tfrac{1}{n-1} \, .
\end{align}
Due to the normalization of $\Mo$ (i.e. $\sum_j \Mo(j)=\id$), the condition \eqref{eq:anti-uni} implies that 
\begin{align}\label{eq:anti-1}
\forall j: \quad \tr{\varrho_j \Mo(j)} = 0 \, ,
\end{align}
hence the condition \eqref{eq:anti-uni} indeed corresponds to the optimal communication matrix $C^{\mathrm{opt}}_n$, defined in \eqref{eq:optimal-n}.

At this point, we recall that quantum states $\varrho_1,\ldots,\varrho_n$ are called \emph{antidistinguishable} \cite{Leifer14,HeKe18,BaJaOpPe14,CaFuSc02,Molina19} if there exists a measurement $\Mo$ with outcomes $1,\ldots,n$ such that \eqref{eq:anti-1} holds. 
Motivated by this, we say that states $\varrho_1,\ldots,\varrho_n$ are \emph{uniformly antidistinguishable} if there exists a measurement $\Mo$ with outcomes $1,\ldots,n$ such that \eqref{eq:anti-uni} holds. 
The question whether Alice and Bob can pass the communication test with $n$ boxes by using a qudit system as a communication medium therefore reduces to the question if there exists a set of $n$ uniformly antidistinguishable qudit states. 
We recall that already a qubit has collections of $n$ antidistinguishable pure states for any $n\geq 2$ \cite{HeKe18}.
This is, however, not enough to solve the task as the uniform antidistinguishability is a stricly stronger condition than antidistinguishability.
In the following example we present a set of antidistinguishable states that are not uniformly antidistinguishable.

\begin{example}
\emph{A collection of states can be antidistinguishable without being uniformly antidistinguishable.} 
To see this, consider four qubit states 
\begin{align*}
\varrho_{1} = \half (\id + \sigma_x) \, , \quad \varrho_{2} =   \half (\id - \sigma_x)\, , \\
\varrho_{3} = \half (\id + \sigma_y)  \, , \quad \varrho_{4} = \half (\id - \sigma_y)  \, . 
\end{align*}
Any measurement $\Mo$ that satisfies \eqref{eq:anti-1} for these states must be of the form
\begin{align*}
\Mo(1) = m_1(\id - \sigma_x) \, , \quad \Mo(2) = m_2 (\id + \sigma_x)\, , \\
\Mo(3) = m_3 (\id - \sigma_y)  \, , \quad \Mo(4) = m_4 (\id + \sigma_y)  \, ,
\end{align*}
for some real numbers $m_j \geq 0$ summing to $1$, and any such choice (e.g. $m_1=m_2=m_3=m_4=\tfrac{1}{4}$) gives a valid solution, thereby showing that the four states are antidistinguishable. 
However, it is not possible to choose the numbers $m_j$ such that the uniform antidistinguishable condition \eqref{eq:anti-uni} is satisfied. 
Namely, if one calculates the outcome probabilities for the states $\varrho_1,\ldots,\varrho_4$ and the measurement $\Mo$, the resulting communication matrix $C$ is as follows:
$$
C=\left[\begin{array}{cccc}
 0 & 2m_2 & m_3 & m_4 \\
2m_1 & 0 & m_3 & m_4 \\
m_1 & m_ 2& 0 & 2m_4 \\
m_1 & m_2 & 2m_3 & 0 \\
\end{array}\right]
$$
It is then clear that there is no choice of $m_1,\ldots,m_4$ that would make $C=C^{\mathrm{opt}}_4$, implying that the states $\varrho_1,\ldots,\varrho_4$ are not uniformly antidistinguishable. 
\end{example}

In the following we demonstrate that there exists a set of uniformly antidistinguishable qudit states for every $n=2,\ldots,d$ but does not exist for any $n > d^2$.

\begin{example}\label{ex:below-d}
From $n$ distinguishable states one can obtain $n$ uniformly antidistinguishable states by forming suitable mixtures. 
In particular, a qudit has collections of $n$ uniformly antidistinguishable states for every $n=2,\ldots,d$.
To see this, let $\varrho_1,\ldots,\varrho_n$ be states such that $\tr{\varrho_j \Mo(k)}=\delta_{jk}$ for some measurement $\Mo$.
We define new states $\varrho'_1,\ldots,\varrho'_n$ as 
\begin{align}
\varrho'_i = \frac{1}{n-1} \sum_{j: j\neq i} \varrho_j \, .
\end{align}
Then
\begin{align}
\tr{\varrho'_i\Mo(j)} = \tfrac{1}{n-1}
\end{align}
for every $i \neq j$, therefore the states $\varrho'_1,\ldots,\varrho'_n$ are uniformly antidistinguishable.
\end{example}

\begin{proposition}\label{prop:max}
There can be at most $d^2$ uniformly antidistinguishable qudit states.
\end{proposition}

\begin{proof}
Let us make a counter assumption that $\varrho_1,\ldots,\varrho_n$, $n>d^2$, are uniformly antidistinguishable, with a measurement $\Mo$. 
We can arrange the states so that  $\{\varrho_i\}_{i=1}^\ell$ forms a basis for $span\{ \varrho_i \}_{i=i}^n$, so that  $\varrho_j = \sum_{i=1}^{\ell} \alpha_{i}^j \rho_i$ for every $1\leq j \leq n$ for some real numbers $\alpha_1^j,\ldots,\alpha_\ell^j$. 
Here $\ell \leq d^2$ as the dimension of $\lhs$ is $d^2$, and therefore $\ell < n$.

Fix $j$ such that $\ell<j \leq n$.
We then have
\begin{align*}
0 = \tr{\varrho_j \Mo(j)} = \sum_{i=1}^{\ell}\alpha_i^j \tr{\rho_i \Mo(j)} = \frac{1}{n-1}\sum_{i=1}^{\ell} \alpha_i^j,
\end{align*}
hence $\sum_{i=1}^{\ell} \alpha_i^j = 0$. 
For every $1\leq k \leq \ell$ we obtain
\begin{align*}
\frac{1}{n-1} &= \tr{\varrho_j \Mo(k)} = \sum_{i=1}^{\ell} \alpha_i^j \tr{\varrho_i \Mo(k)} = \frac{1}{n-1}(\sum_{i =1}^{\ell}\alpha_i^j - \alpha_k^j) \\
&= -\frac{1}{n-1}\alpha_k^j \, , 
\end{align*}
implying that $\alpha_k^j = -1$. This is a contradiction and therefore there cannot be more than $d^2$ uniformly antidistinguishable qudit states.
\end{proof}

The remaining question is whether a set of $n$ uniformly antidistinguishable qudit states exists for $d<n \leq d^2$.
We give a partial answer by relating this question to a suitable symmetry property of the respective measurement.

\begin{example}\label{ex:sic}
A \emph{symmetric informationally complete} (SIC) measurement is a measurement $\Mo$ with $d^2$ outcomes such that the linear span of $\ran \Mo$ is $\lh$, each operator $\Mo(j)$ is rank-1, $\tr{\Mo(j)}$ is constant for all $j$, and $\tr{\Mo(j)\Mo(k)}$ is constant for all $j\neq k$ \cite{Zauner99,ReBlScCa04}.
It can be shown that from this definition follows that for all $j\neq k$ we have
\begin{align}
& \tr{\Mo(j)}=1/d \, , \label{eq:SIC-1}\\
& \tr{\Mo(j)\Mo(k)}=1/(d^2(d+1)) \, . \label{eq:SIC-2}
\end{align}
It further follows that
\begin{align}
\tr{\Mo(j)^2}=\tr{\Mo(j)}^2=1/d^2  \, .\label{eq:SIC-3}
\end{align}
Let $\Mo$ be a fixed SIC measurement. 
For each $i=1,\ldots,d^2$, we define a state $\varrho_i$ as
\begin{equation}\label{eq:SIC-states}
\varrho_i = \tfrac{1}{d-1} (\id - d \Mo(i)) \, .
\end{equation}
Using the defining conditions \eqref{eq:SIC-1} and \eqref{eq:SIC-2}, we obtain
\begin{equation}
\tr{\varrho_i \Mo(i)}= \tfrac{1}{d-1} (\tr{\Mo(i)} - d \tr{\Mo(i)^2}) =0 
\end{equation}
and
\begin{equation}
\tr{\varrho_i \Mo(j)}= \tfrac{1}{d-1} (\tr{\Mo(j)} - d \tr{\Mo(i)\Mo(j)}) = \tfrac{1}{d^2-1} \, .
\end{equation}
In conclusion, the $d^2$ states $\varrho_1, \ldots,\varrho_{d^2}$ defined in \eqref{eq:SIC-states} are uniformly antidistinguishable. 
It is an open problem if a SIC measurement exists in every dimension $d$, but analytic and numerical solutions are known up to $d\sim 100$; see \cite{ApChFlWa18} and the references given there.
\end{example}

\begin{example}\label{ex:sym}
As a modification of the construction in the previous example, let us suppose that $\Mo$ is a measurement with $n \geq d$ outcomes such that each operator $\Mo(j)$ is rank-1, $\tr{\Mo(j)}$ is constant for all $j$, and $\tr{\Mo(j)\Mo(k)}$ is constant for all $j\neq k$.
We call this kind of measurement \emph{symmetric}.  
The defining requirements imply that 
\begin{align}
& \tr{\Mo(j)}= d/n \, , \label{eq:sym-1} \\
& \tr{\Mo(j)\Mo(k)}= \frac{dn-d^2}{n^2(n-1)} \, . \label{eq:sym-2}
\end{align}
For each $i=1,\ldots,n$, we define a state $\varrho_i$ as
\begin{equation}\label{eq:sym-states}
\varrho_i =  \tfrac{1}{d-1}( \id -  \tfrac{n}{d} \Mo(i)) .
\end{equation}
Using the defining conditions \eqref{eq:sym-1} and \eqref{eq:sym-2}, we obtain
\begin{equation}
\tr{\varrho_i \Mo(i)} = \tfrac{1}{d-1} (\tr{\Mo(i)} - \tfrac{n}{d} \tr{\Mo(i)^2}) = 0
\end{equation}
and
\begin{align*}
\tr{\varrho_i \Mo(j)} & = \tfrac{1}{d-1} (\tr{\Mo(j)} - \tfrac{n}{d} \tr{\Mo(i)\Mo(j)})  = \tfrac{1}{n-1} \, . 
\end{align*}
In conclusion, the $n$ states $\varrho_1, \ldots,\varrho_{n}$ are uniformly antidistinguishable. 
\end{example}

The previous results allow us to draw the following conclusion.

\begin{theorem}\label{thm:qubit}
There exists a set of $n$ uniformly antidistinguishable qubit states if and only if $n$ is $2,3$ or $4$.
\end{theorem}

The proof of this statement follows from our earlier observations.
By Prop. \ref{prop:max} there can be at most four antidistinguishable qubit states and by Example \ref{ex:below-d} there is a set of two antidistinguishable qubit states.
We use the method of Examples \ref{ex:sic} and \ref{ex:sym} to construct uniformly antidistinguishable sets of $3$ and $4$ states.
First, a qubit SIC measurement is
\begin{align*}
\Mo(1) &= \tfrac{1}{4} (\id + \tfrac{1}{\sqrt 3} (\sigma_x + \sigma_y + \sigma_z)) \, , \quad \Mo(2) = \tfrac{1}{4} (\id + \tfrac{1}{\sqrt 3} (\sigma_x - \sigma_y - \sigma_z))\, , \\
\Mo(3) &= \tfrac{1}{4} (\id + \tfrac{1}{\sqrt 3} (-\sigma_x + \sigma_y - \sigma_z)) \, , \quad \Mo(4) = \tfrac{1}{4} (\id + \tfrac{1}{\sqrt 3} (-\sigma_x - \sigma_y + \sigma_z))
\end{align*}
and this leads to four uniformly antidistinguishable states
\begin{align*}
\varrho_{1} &= \half (\id - \frac{1}{\sqrt 3}(\sigma_x + \sigma_y + \sigma_z)) \, , \quad \varrho_{2} =   \half (\id - \frac{1}{\sqrt 3}(\sigma_x - \sigma_y - \sigma_z))\, , \\
\varrho_{3} &= \half (\id - \frac{1}{\sqrt 3}(-\sigma_x + \sigma_y - \sigma_z))  \, , \quad \varrho_{4} = \half (\id - \frac{1}{\sqrt 3}(-\sigma_x - \sigma_y + \sigma_z) )  \, . 
\end{align*}
For $n=3$ we define a measurement $\No$ as
\begin{align*}
\No(1) &= \tfrac{1}{3} (\id +\sigma_x) \, , \quad \No(2) = \tfrac{1}{3} (\id - \tfrac{1}{2} \sigma_x + \tfrac{\sqrt{3}}{2} \sigma_y )\, , \\
\No(3) &= \tfrac{1}{3} (\id - \tfrac{1}{2} \sigma_x - \tfrac{\sqrt{3}}{2} \sigma_y )   \, . 
\end{align*}
It is straightforward to verify that $\No$ is symmetric, thereby giving rise to three uniformly antidistinguishable qubit states. 
The obtained states are
\begin{align*}
\varrho_{1} &= \half (\id - \sigma_x) \, , \quad \varrho_{2} =   \half (\id + \tfrac 12 \sigma_x - \tfrac{\sqrt 3}{2} \sigma_y )\, , \\
\varrho_{3} &= \half (\id + \tfrac 12 \sigma_x + \tfrac{\sqrt 3}{2} \sigma_y)    \, . 
\end{align*}

\begin{figure}
\centering
\begin{subfigure}{.5\textwidth}
\centering
\includegraphics[scale=.5]{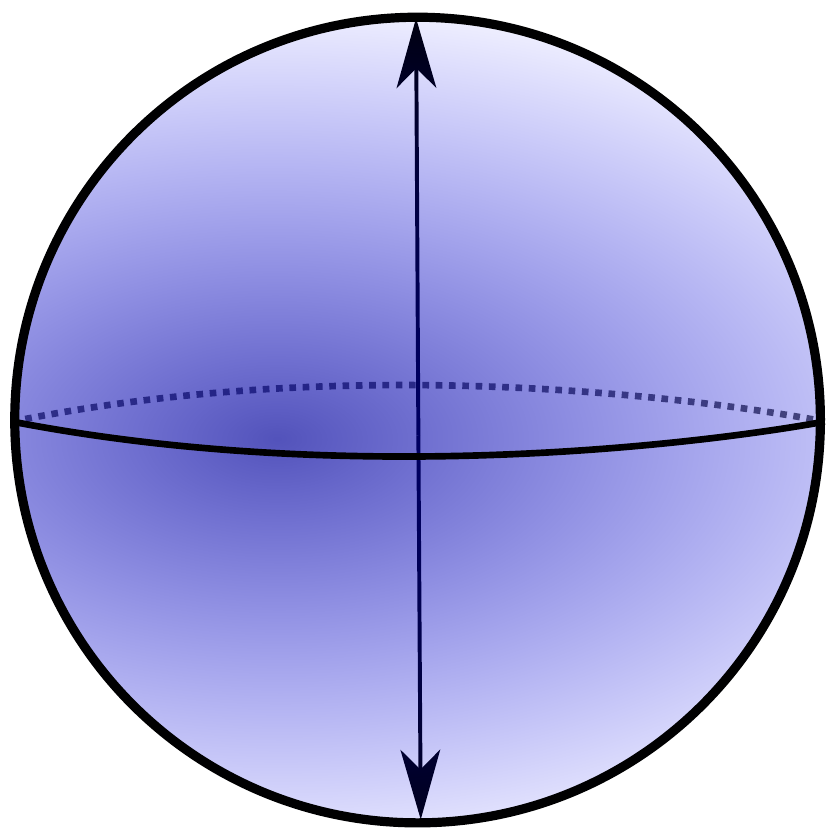}
\caption{$C^{\mathrm{opt}}_{2}$}
\label{fig:qubit_2-1}
\end{subfigure}%
\begin{subfigure}{.5\textwidth}
\centering
\includegraphics[scale=.5]{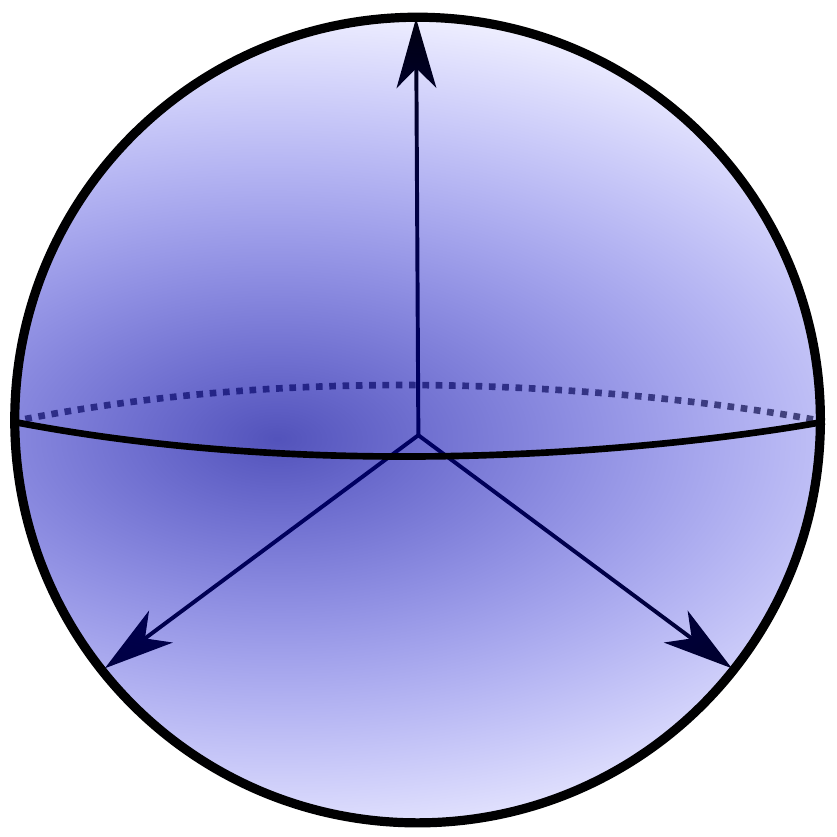}
\caption{$C^{\mathrm{opt}}_{3}$}
\label{fig:qubit_3-1}
\end{subfigure}
\begin{subfigure}{.5\textwidth}
\centering
\includegraphics[scale=.5]{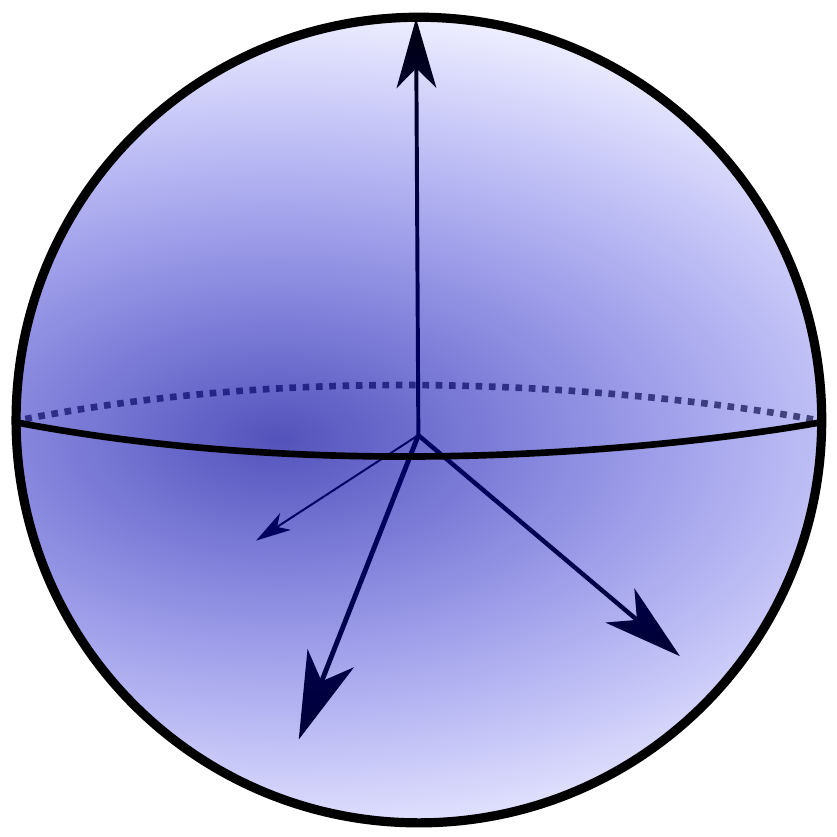}
\caption{$C^{\mathrm{opt}}_{4}$}
\label{fig:qubit_4-1}
\end{subfigure}
\caption{The sets of pure qubit states in the Bloch ball that implement the corresponding optimal communication matrices. 
The points in the ball form (a) a line segment, (b) a triangle and (c) a tetrahedron.
}
\label{fig:Bloch_sphere}
\end{figure}

With an additional inspection we can see that any uniformly antidistinguishable set of qubit states is a unitary transformation of one of the three previously listed sets.
This is easy to see in the Bloch ball picture, which leads to an alternative proof of Theorem \ref{thm:qubit}.
Let us first recall that antidistinguishable qubit states must be pure \cite{HeKe18}, namely, a mixed qubit state $\varrho$ gives $\tr{\varrho E}>0$ for any nonzero positive operator $E$, hence making the condition \eqref{eq:anti-1} impossible to satisfy.
We are thus considering $n$ pure qubit states $\varrho_1,\ldots,\varrho_n$, and we write them in the Bloch representation as $\varrho_j = \half (\id + \vr_j \cdot \vsigma)$, where $\vr_j \in \R^3$, $\no{\vr_j}=1$.
This set of states is antidistinguishable if and only if $\sum_j t_j\vr_j = 0$ for some real numbers $t_j>0$ satisfying $\sum_j t_j =2$ \cite{HeKe18}.
In this case, the antidistinguishing measurement is unique and given as 
\begin{equation}
\Mo(j)=\frac{t_j}{2}(\id - \vr_j \cdot \vsigma) \, .
\end{equation}
A direct calculation using this explicit form of $\Mo$ shows that the stronger condition \eqref{eq:anti-uni} of uniform antidistinguishability is equivalent to
\begin{align}\label{eq:qubit-uni}
\forall j\neq k: \quad t_k( 1-\vr_j \cdot \vr_k) = \frac{2}{n-1}  \, .
\end{align}
This condition implies that 
\begin{align}
t_k( 1-\vr_j \cdot \vr_k) =  t_j( 1-\vr_k \cdot \vr_j) 
\end{align}
for all $j\neq k$.
But since $\vr_j \cdot \vr_k = \vr_k \cdot \vr_j$, this means that either $t_j=t_k$ or $\vr_j \cdot \vr_k=1$.
The latter equality cannot hold as $\vr_j \neq \vr_k$.
Therefore, we conclude that $t_1 = t_2 = \cdots = t_n$, implying that $t_1 = \cdots = t_n=2/n$.
Using this fact the condition \eqref{eq:qubit-uni} takes the form
\begin{align}\label{eq:qubit-sym-23}
\forall j\neq k: \quad \vr_j \cdot \vr_k = \frac{1}{1-n}  \, .
\end{align}
In particular, all the inner products between two different Bloch vectors must be equal and fixed by $n$.
For a chosen integer $n=2,3,4$, the set of states is therefore unique up to a unitary transformation.
These constellations are depicted in Fig. \ref{fig:Bloch_sphere}.

%%%%%%%%%%%%%%%%%%%%%%
\section{Generalized communication tests}\label{sec:gen}
%%%%%%%%%%%%%%%%%%%%%%

An obvious generalization of the previous tests is that Charlie reveals $t$ of the $n-1$ empty boxes to Alice, where $t$ is a fixed integer between $1$ and $n-1$.
The number of revealed boxes is agreed before the test starts and Alice and Bob can choose their encoding and decoding accordingly.
The first test explained in the beginning of Sec. \ref{sec:ignorance} corresponds to $t=n-1$ and then Alice knows exactly where the candy is, while the test described after, the one that connected to $C^{\mathrm{opt}}_n$, corresponds to $t=1$. 
 The success probability for finding the candy for Alice alone is $\tfrac{1}{n-t}$, and we want to know if Alice and Bob can reach this success probability when their communication is limited to some physical system, e.g. qudit, from Alice to Bob.

For Alice to be able to send  to Bob the full information that she learns from Charlie, they must use an encoding that has $\binom{n}{t}$ labels, one label for each possible $t$ tuple.
It is convenient to use the ordered tuples $(i_1,\ldots,i_t)$, $1\leq i_1 < \cdots < i_t \leq n$, and we denote by $\Omega^{<}_{n;t}$ the set of all these ordered tuples. 
In this notation, the information that Alice gets from Charlie is an order tuple $(i_1,\ldots,i_t)\in\Omega^{<}_{n;t}$, specifying the boxes which Charlie will keep empty.
The condition for the optimal communication is therefore
\begin{equation}\label{eq:popt-nt}
  p(j|(i_1,\ldots,i_t)) =\begin{cases}
    0, & \text{if $i\in \{ i_1,\ldots,i_t \} $}.\\
    \tfrac{1}{n-t}, & \text{otherwise}.
  \end{cases}
\end{equation}
The reason for this optimality condition is analogous to the discussion of Sec. \ref{sec:ignorance}; Bob must avoid the known empty locations and choose from other locations with the uniform probability.

To write the optimal communication probabilities as a matrix, we use the lexicographic ordering in $\Omega^{<}_{n;t}$ and label the elements of $\Omega^{<}_{n;t}$ by integers from $1$ to $\binom{n}{t}$. 
For instance, in the case of $n=4$, $t=2$, we use indices
\begin{align*}
1 \leftrightarrow (1,2) \qquad 2 \leftrightarrow (1,3) \qquad 3 \leftrightarrow (1,4) \\
4 \leftrightarrow (2,3) \qquad 5 \leftrightarrow (2,4) \qquad 6 \leftrightarrow (3,4) \\
\end{align*}
We denote by $C^{\mathrm{opt}}_{n,t}$ the optimal communication matrix of the test, thereby having elements $[C^{\mathrm{opt}}_{n,t}]_{ij} = p(j|i)$, where $p$ is the conditional probability distribution of \eqref{eq:popt-nt} written in the previously explained indexing. 
For instance, 
$$
C^{\mathrm{opt}}_{4,2} = \frac{1}{2} \left[\begin{array}{cccc} 0 & 0 & 1 & 1\\ 0 & 1 & 0 & 1\\ 0 & 1 & 1 & 0\\ 1 & 0 & 0 & 1\\ 1 & 0 & 1 & 0\\ 1 & 1 & 0 & 0  \end{array}\right] \, .
$$
Generally, $C^{\mathrm{opt}}_{n,t}$ is $\binom{n}{t} \times n$ matrix and it can be written as follows. 
The first row is
\begin{equation}\label{eq:firstrow}
\tfrac{1}{n-t} \left[\begin{array}{cccccc} 0 & \cdots & 0 & 1 & \cdots 1  \end{array}\right]
\end{equation}
with $t$ zeros.
The other rows are all possible permutations of this row that give a different row.
The order of the rows is irrelevant and we come back to that point in Sec. \ref{sec:ultraweak}.
This notation generalizes the earlier notation for $C^{\mathrm{opt}}_{n}$ introduced in Sec. \ref{sec:ignorance} and we have $C^{\mathrm{opt}}_{n}\equiv C^{\mathrm{opt}}_{n,1}$.
 
A qudit implementation of $C^{\mathrm{opt}}_{n,t}$ means that there exist $k=\binom{n}{t}$ qudit states $\varrho_{1},\ldots,\varrho_{k}$ and a measurement $\Mo$ with the outcome set $\{1,\ldots,n \}$ such that
\begin{align}\label{eq:cntopt}
[C^{\mathrm{opt}}_{n,t}]_{ij} = \tr{\varrho_{i}\Mo(j)} \, .
\end{align}
for all $i=1,\ldots,k$ and $j=1,\ldots,n$.

Our first observation is that, other than the cases listed in Theorem \ref{thm:qubit}, the previously described tests cannot be implemented with a qubit system. 

\begin{proposition}\label{prop:noqubit}
The optimal communication matrix $C^{\mathrm{opt}}_{n,t}$ has no qubit realization for any $t\geq 2$.
\end{proposition}

\begin{proof}
Let $t\geq 2$, so that $n\geq t+1=3$.
We assume that $C^{\mathrm{opt}}_{n,t}$ is given by qubit states $\varrho_{1},\ldots,\varrho_{k}$ and a qubit measurement $\Mo$ via \eqref{eq:cntopt}.
We first observe that the rows of $C^{\mathrm{opt}}_{n,t}$ are all different, hence the states $\varrho_{1},\ldots,\varrho_{k}$ must be different.
Since $t\geq 2$, the first two rows of $C^{\mathrm{opt}}_{n,t}$ are of the form
\begin{equation}\label{eq:tworow}
\tfrac{1}{n-t} \left[\begin{array}{cccc} 0 & * & * \cdots & *  \end{array}\right]
\end{equation}
Therefore, we have $\tr{\varrho_1 \Mo(1)} = \tr{\varrho_2 \Mo(1)}=0$. 
We further have $\tr{\varrho_k \Mo(1)}\neq 0$ and hence $\Mo(1)\neq 0$.
But since a nonzero singular operator on $\C^2$ has a one-dimensional kernel, this is a contradiction.
\end{proof}

Our second observation is that $C^{\mathrm{opt}}_{4,2}$ has a qutrit implementation.   

\begin{example}
\emph{The optimal communication matrix $C^{\mathrm{opt}}_{4,2}$ has a qutrit realization.}
We choose the following six pure states: 
\begin{align*}
 & \rho_{12} = \begin{bmatrix}
1 & 0 & 0 \\
0 & 0 & 0 \\
0 & 0 & 0
\end{bmatrix}, \, \rho_{13} = \begin{bmatrix}
\frac 14 & -\frac 14 & - \frac{1}{2\sqrt2} \\[0.3em]
- \frac 14 & \frac 14 & \frac{1}{2\sqrt2} \\[0.3em]
-\frac{1}{2\sqrt2} & \frac{1}{2\sqrt2} & \frac 12
\end{bmatrix}, \, \rho_{14} = \begin{bmatrix}
\frac 14 & \frac 14 & - \frac{1}{2\sqrt2} \\[0.3em]
 \frac 14 & \frac 14 & -\frac{1}{2\sqrt2} \\[0.3em]
-\frac{1}{2\sqrt2} & -\frac{1}{2\sqrt2} & \frac 12
\end{bmatrix}, \\[0.5em]
& \rho_{23} = \begin{bmatrix}
\frac 14 & \frac 14 &  \frac{1}{2\sqrt2} \\[0.3em]
 \frac 14 & \frac 14 & \frac{1}{2\sqrt2} \\[0.3em]
\frac{1}{2\sqrt2} & \frac{1}{2\sqrt2} & \frac 12
\end{bmatrix}, \, \rho_{24} = \begin{bmatrix}
\frac 14 & -\frac 14 &  \frac{1}{2\sqrt2} \\[0.3em]
- \frac 14 & \frac 14 & -\frac{1}{2\sqrt2} \\[0.3em]
\frac{1}{2\sqrt2} & -\frac{1}{2\sqrt2} & \frac 12
\end{bmatrix}, \, \rho_{34} = \begin{bmatrix}
0 & 0 & 0 \\
0 & 1 & 0 \\
0 & 0 & 0
\end{bmatrix}.
\end{align*}
The measurement $\Mo$ that leads to the communication matrix $C^{\mathrm{opt}}_{4,2}$ with the previous states is the following: 
\begin{align*}
& \Mo(1) = \begin{bmatrix}
\frac 12 & 0 & -\frac{1}{2\sqrt2} \\[0.3em]
0 & 0 & 0 \\[0.3em]
-\frac{1}{2\sqrt2} & 0 & \frac 14
\end{bmatrix}, \, \Mo(2) = \begin{bmatrix}
\frac 12 & 0 & \frac{1}{2\sqrt2} \\[0.3em]
0 & 0 & 0 \\[0.3em]
\frac{1}{2\sqrt2} & 0 & \frac 14
\end{bmatrix}, \\[0.5em]
& \Mo(3) = \begin{bmatrix}
0 & 0 & 0\\[0.3em]
0 & \frac 12 & \frac{1}{2\sqrt2} \\[0.3em]
0 & \frac{1}{2\sqrt2} & \frac 14
\end{bmatrix}, \, \Mo(4) = \begin{bmatrix}
0 & 0 & 0\\[0.3em]
0 & \frac 12 & -\frac{1}{2\sqrt2} \\[0.3em]
0 & -\frac{1}{2\sqrt2} & \frac 14
\end{bmatrix}.
\end{align*}
It is straightforward to check that $\tr{\varrho_{ij}\Mo(k)} = \frac 12$ if $k \in \{i,j\}$, and zero otherwise.
\end{example}

We conclude that the new communication tests introduced in this section, i.e. tests with $1<t<n-1$, cannot be reduced to the earlier tests. 
The communication tests have an operationally motivated hierarchy and we will discuss that in Sec. \ref{sec:ultraweak}.

%%%%%%%%%%%%%%%%%%%%%%
\section{Communication tests using general physical systems}\label{sec:systems}
%%%%%%%%%%%%%%%%%%%%%%

So far, we have been discussing if a certain communication matrix $C$ has a qudit implementation, meaning that $C_{ij}=\tr{\varrho_i \Mo(j)}$ for some qudit states $\varrho_i$ and qudit measurement $\Mo$.
We can equally well ask if a communication matrix $C$ has an implementation with some other physical communication medium.
It is also interesting to compare qubit with some hypothetical physical systems, such as real vector space qubit.

As a physical system we will understand anything that is specified with a convex subset $\mathcal{S}$ of a real vector space $\mathcal{V}$, the set $\mathcal{S}$ corresponding to possible states of the system.
a measurement with $n$ outcomes is an affine map $M:\mathcal{S}\to\Delta_n$, where $\Delta_n$ is the set of probability distributions $p=(p_1,\ldots,p_n)$. 
We denote by $\mathcal{M}$ the set of all measurements. 
A physical system is identified with the state space $\mathcal{S}$.
This is exactly the framework of operational probabilistic theories, extensively used in studies of quantum foundations; see e.g. \cite{KiNuIm10,MaMu11} and references therein.

We say that a communication matrix $C$ of size $k \times n$ can be implemented with a physical system $\mathcal{S}$ if there are states $s_1,\ldots,s_k$ and a measurement $M$ with the outcomes $1,\ldots,n$ such that 
$C_{ij} = [M(s_i)]_j $ for all $i,j$.
The operational dimension of $\mathcal{S}$ is, by the definition, the maximal number of distinguishable states. 
In our framework, this is to say that it is the maximal number $n$ such that the identity matrix $\id_n$ can be implemented with $\mathcal{S}$ but $\id_{n+1}$ cannot.

Given a physical system $\mathcal{S}$, an obvious task is to characterize all pairs $(n,t)$ such that the optimal communication matrix $C^{\mathrm{opt}}_{n,t}$ can be implemented with $\mathcal{S}$.
The implementable pairs can be represented in the form of a table; our earlier results on the qubit are summarized in Fig. \ref{fig:comm_tab_qubits}.
We call these communication tables of the respective physical systems.

With the following example we demonstrate that two physical systems with the same operational dimension can have different communication tables.

\begin{example}
The real vector space qubit, or rebit for short, is the two-dimensional system whose Hilbert space is defined over the real numbers. In particular, the state space of rebits is spanned by the matrices $\id$, $\sigma_x$ and $\sigma_z$. Hence, the linear span of the Bloch vectors of rebits span a two-dimensional subspace of $\mathbb{R}^3$.
The operational dimension of the real bit is two, hence $C^{\mathrm{opt}}_{2,1}$ is implementable but $C^{\mathrm{opt}}_{n,n-1}$ for $n\geq 3$ are not.
An analogous discussion as we had in the case of qubit in Sec. \ref{sec:qudit} and Prop. \ref{prop:noqubit} shows that one can implement $C^{\mathrm{opt}}_{2,1}$ and $C^{\mathrm{opt}}_{3,1}$, but not any other optimal communication matrices. 
The table summarizing the implementable communication tests of the rebit is depicted in Fig. \ref{fig:comm_tab_rebits}.
\end{example}

\begin{figure}
\centering
\begin{subfigure}{.5\textwidth}
\centering
\scalebox{1.2}{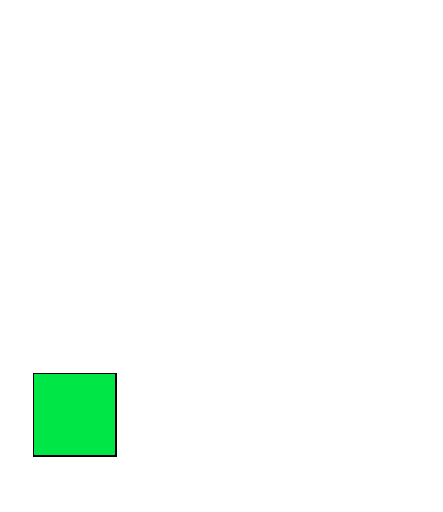}
\caption{qubit}
\label{fig:comm_tab_qubits}
\end{subfigure}%
\begin{subfigure}{.5\textwidth}
\centering
\scalebox{1.2}{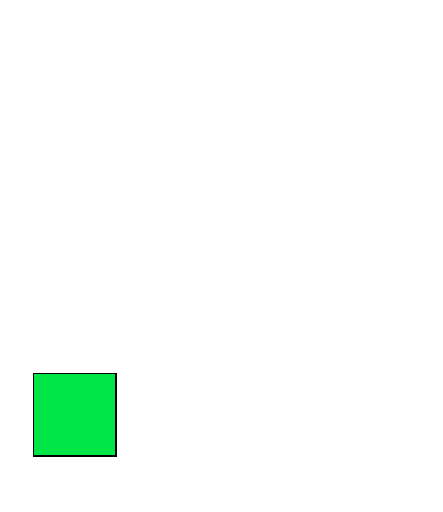}
\caption{rebit}
\label{fig:comm_tab_rebits}
\end{subfigure}
\caption{The implementable communication tests for the qubit and rebit.}
\label{fig:comm_tab}
\end{figure}

%%%%%%%%%%%%%%%%%%%%%%
\section{Ultraweak matrix majorization}\label{sec:ultraweak}
%%%%%%%%%%%%%%%%%%%%%%

We have introduced communication tasks and the relevant optimal communication matrices $C^{\mathrm{opt}}_{n,t}$ for $n\geq 2$ and $1\leq t \leq n-1$.
The question then arises if some of these tasks are more difficult than others. 
In this section we explain that the proper ordering of difficulty is given by the ultraweak matrix majorization, a preordering that we will introduce next.

We denote by $\M{a}{b}$ the set of $a\times b$ real matrices and by $\Mr{a}{b}$ the set of $a\times b$ row-stochastic matrices.
Further, we denote by $\Mall$ the set of all real matrices of finite size, i.e., $\Mall = \cup_{a,b} \M{a}{b}$.

\begin{definition}\label{def:majorization}
For two matrices $M\in\M{a}{b}$ and $N\in\M{c}{d}$, we denote $M\uleq N$ if there are row-stochastic matrices $L\in\Mr{a}{c}$ and $R\in\Mr{d}{b}$ such that $M=LNR$.
We then say that $M$ is \emph{ultraweakly majorized} by $N$.
We further denote $M \simeq N$ if $M\uleq N \uleq M$.
\end{definition}

This definition is related to the concepts of matrix majorization and weak matrix majorization; see \cite{Dahl99,PeMaSi05}.
Those relations are typically defined for matrices of the same size, and then the matrix majorization corresponds to having $L=\id$ in Def. \ref{def:majorization} while the weak matrix majorization corresponds to having $R=\id$.
Hence, if a matrix $M$ is majorized or weakly majorized by another matrix $N$, then $M$ is also ultraweakly majorized by $N$.

We start by listing three sufficient criteria for two matrices being ultraweakly equivalent.
The proofs are straightforward and we omit them.

\begin{proposition}\label{prop:equivalent}
In the following cases two matrices $M$ and $N$ satisfy $M\simeq N$.
\begin{itemize}
\item[(a)] $M,N\in\M{a}{b}$ and $N$ is obtained from $M$ by permuting rows and columns.
\item[(b)] $M\in\M{a}{b}$, $N\in\M{c}{b}$, $c>a$, and $N$ is obtained from $M$ by duplicating some of the rows of $M$.
\item[(c)] $M\in\M{a}{b}$, $N\in\M{a}{c}$, $c>a$, and $N$ is obtained from $M$ by adding zero columns to $M$.
\end{itemize}
\end{proposition}

\begin{proposition}
For any $M\in\Mr{a}{n} \cup \Mr{n}{b}$, we have 
$$
V_n \uleq M \uleq \id_n \, , 
$$
where 
$$
V_n = \frac{1}{n} \left[\begin{array}{ccc} 1 & \cdots & 1 \\ \vdots &  \ddots & \vdots  \\ 1 & \cdots & 1\end{array}\right] \, .
$$
\end{proposition}

\begin{proof}
Let us first consider $M \uleq \id_n$. 
If $M \in \Mr{a}{n}$, then we can choose $R = \id_n$ and $L = M$ and now $L\id_n R = M$. 
On the other hand if $M \in \Mr{n}{b}$, then we can choose $L = \id_n$ and $R = M$. 
This proves the first part.

Let us then consider $V_n \uleq M$. 
If $M \in \Mr{a}{n}$, we can choose $L = V_{n,a}$ and $R = V_{n}$, where $V_{n,a}$ is a $n \times a $ row-stochastic matrix with all elements equal. Then $LMR = V_n$. 
Likewise for $M \in \Mr{n}{b}$, we can choose $L = V_n$ and $R = V_{b,n}$. This concludes the proof.
\end{proof}

Next we provide some nontrivial examples when ultraweak majorization holds and when it doesn't for the optimal communication matrices defined in Sec. \ref{sec:gen}.

\begin{example}\label{sec:uwm-examples}
With the following matrix equations we demonstrate that  $C^{\mathrm{opt}}_{2,1} \uleq C^{\mathrm{opt}}_{4,2}$, $C^{\mathrm{opt}}_{3,1} \uleq C^{\mathrm{opt}}_{4,2}$ and $C^{\mathrm{opt}}_{4,1} \uleq C^{\mathrm{opt}}_{4,2}$:
\begin{align}\label{eq:2142}
\left[\begin{array}{cccccc}1 & 0 & 0 & 0 & 0 & 0 \\
0 & 0 & 0 & 0 & 0 & 1 \\
\end{array}\right]\frac12 \left[\begin{array}{cccc} 0 & 0 & 1 & 1\\ 0 & 1 & 0 & 1\\ 0 & 1 & 1 & 0\\ 1 & 0 & 0 & 1\\ 1 & 0 & 1 & 0\\ 1 & 1 & 0 & 0  \end{array}\right] \left[\begin{array}{cc} 1 & 0 \\ 1 & 0 \\ 0 & 1 \\ 0 & 1\end{array}\right] =  \left[ \begin{array}{cc}
0 & 1 \\ 1 & 0 
\end{array}\right] \, ,
\end{align}

\begin{align}
\left[\begin{array}{cccccc}1 & 0 & 0 & 0 & 0 & 0 \\
0 & 1 & 0 & 0 & 0 & 0 \\
0 & 0 & 1 & 0 & 0 & 0
\end{array}\right]\frac12 \left[\begin{array}{cccc} 0 & 0 & 1 & 1\\ 0 & 1 & 0 & 1\\ 0 & 1 & 1 & 0\\ 1 & 0 & 0 & 1\\ 1 & 0 & 1 & 0\\ 1 & 1 & 0 & 0  \end{array}\right] \left[\begin{array}{ccc} 1 & 0 & 0 \\ 1 & 0 & 0 \\ 0 & 1 & 0 \\ 0 & 0 & 1  \end{array}\right] = \frac12 \left[ \begin{array}{ccc}
0 & 1 & 1 \\ 1 & 0 & 1 \\ 1 & 1 & 0
\end{array}\right] \, ,
\end{align}

\begin{align}
\frac{1}{3} \left[\begin{array}{cccccc} 1 & 1 & 1 & 0 & 0 & 0\\ 1 & 0 & 0 & 1 & 1 & 0\\ 0 & 1 & 0 & 1 & 0 & 1\\ 0 & 0 & 1 & 0 & 1 & 1 \end{array}\right]\frac{1}{2} \left[\begin{array}{cccc} 0 & 0 & 1 & 1\\ 0 & 1 & 0 & 1\\ 0 & 1 & 1 & 0\\ 1 & 0 & 0 & 1\\ 1 & 0 & 1 & 0\\ 1 & 1 & 0 & 0  \end{array}\right] \id_4 = \frac{1}{3} \left[\begin{array}{cccc} 0 & 1 & 1 & 1\\ 1 & 0 & 1 & 1\\ 1 & 1 & 0 & 1\\ 1 & 1 & 1 & 0  \end{array}\right] \, .
\end{align}
\end{example}

\begin{example}\label{ex:2131}
We have both $C^{\mathrm{opt}}_{2,1} \nuleq C^{\mathrm{opt}}_{3,1}$ and  $C^{\mathrm{opt}}_{3,1} \nuleq C^{\mathrm{opt}}_{2,1}$.
Firstly, observe that $\rank{C^{\mathrm{opt}}_{3,1}} = 3 $ and $\rank{C^{\mathrm{opt}}_{2,1}} = 2 $, hence we cannot have $C^{\mathrm{opt}}_{3,1} \uleq C^{\mathrm{opt}}_{2,1}$ as the matrix rank cannot increase in matrix multiplication. 
On the other hand, suppose that there are matrices $L$ and $R$ such that
 \begin{align*}
C^{\mathrm{opt}}_{2,1}=L C^{\mathrm{opt}}_{3,1} R &= \begin{bmatrix}
a_{11} & a_{12} & a_{13} \\ a_{21} & a_{22} & a_{23}
\end{bmatrix}\frac 12 \begin{bmatrix}
0 & 1 & 1 \\ 1 & 0 & 1\\1 & 1 & 0
\end{bmatrix}\begin{bmatrix}
b_{11} & b_{12} \\b_{21} & b_{22} \\b_{31} & b_{32} \\
\end{bmatrix} 
\end{align*}
This matrix equation reduces to the following set of equations: \begin{align*}
a_{11}(b_{21}+b_{31})+a_{12}(b_{11}+b_{31})+a_{13}(b_{11}+b_{21}) &= 0 \\
a_{11}(b_{22}+b_{32})+a_{12}(b_{12}+b_{32})+a_{13}(b_{12}+b_{22}) &= 2 \\
a_{21}(b_{21}+b_{31})+a_{22}(b_{11}+b_{31})+a_{23}(b_{11}+b_{21}) &= 2 \\
a_{21}(b_{22}+b_{32})+a_{22}(b_{12}+b_{32})+a_{23}(b_{12}+b_{22}) &= 0 .
\end{align*}
It is straightforward to check that this set of equations does not have a solution where both $L$ and $R$ are row-stochastic. For instance, if $a_{11}=a_{12} = 0$, then $b_{11}=b_{21}=0$ and $a_{13}=1$. It follows that $b_{12}=b_{22} = 1$ and $(a_{21}+a_{22})b_{31}=2$, which is impossible. If $a_{11} = 0 $ while $a_{12},a_{13} \neq 0$, then $b_{11} = b_{21} = b_{31} = 0 $, which contradicts the third equation. Moreover, $a_{11}$, $a_{12}$ and $a_{13}$ cannot all be non-zero for the same reason. Similar observations for other choices of zeros for the first row of $L$ show that there is no choice that works, and therefore $C^{\mathrm{opt}}_{2,1} \nuleq C^{\mathrm{opt}}_{3,1}$.
\end{example}

It is straightforward to see that the ultraweak majorization is a reflexive and transitive relation on $\Mall$, hence a preorder. 
Before going into mathematical properties of the ultraweak majorization, we explain its connection to the communication tasks of the earlier sections.

\begin{proposition}\label{prop:processing}
Let $C\in\Mr{a}{b}$ be a row-stochastic matrix that has an implementation with a physical system $\mathcal{S}$ (in the sense explained in Sec. \ref{sec:systems}). 
Then any row-stochastic matrix $C'\in\Mr{c}{d}$ satisfying $C'\uleq C$ has also an implementation with $\mathcal{S}$.
\end{proposition}

\begin{proof}
The assumption that $C$ can be implemented with a physical system $\mathcal{S}$ means that there are states $s_1,\ldots,s_a$ and a measurement $M$ with the outcomes $1,\ldots,b$ such that 
$C_{ij} = [M(s_i)]_j $ for all $i,j$.
Suppose that $C'\in\Mr{c}{d}$ satisfies $C'\uleq C$.
This means that there exist $L\in\Mr{c}{a}$ and $R\in\Mr{b}{d}$ such that $C'=LCR$.
We define new states $s'_1,\ldots,s'_c$ as
\begin{align*}
s'_p = \sum_{i=1}^a L_{pi} s_i
\end{align*}
and a new measurement $M'$ as
\begin{align*}
[M'(s)]_q = \sum_{j=1}^b R_{jq} [M(s)]_j \, .
\end{align*}
Then $[M'(s'_p)]_q = C'_{pq}$.
\end{proof}

To illustrate the content of Prop. \ref{prop:processing} and its proof, we consider the relation $C^{\mathrm{opt}}_{2,1} \uleq C^{\mathrm{opt}}_{4,2}$ that was shown in \eqref{eq:2142}.
Suppose that we have an implementation for $C^{\mathrm{opt}}_{4,2}$.
The matrix $L$ corresponds to a relabeling of the original states and in this specific case it means that we choose $s'_1=s_1$ and $s'_2=s_6$, never using the other states. 
The matrix $R$ corresponds to a post-processing of measuring outcomes and in this specific case it means that the outcomes $1$ and $2$ are interpreted as $1$ while the outcomes $3$ and $4$ are interpreted as $2$.
With these relabelings the original implementation of $C^{\mathrm{opt}}_{4,2}$ becomes an implementation of $C^{\mathrm{opt}}_{2,1}$.
Generally, relabelings on both sides can be stochastic and an illustration of this kind of scenario is given in Fig. \ref{fig:ultraweak}.

\begin{figure}
\includegraphics[width=8cm]{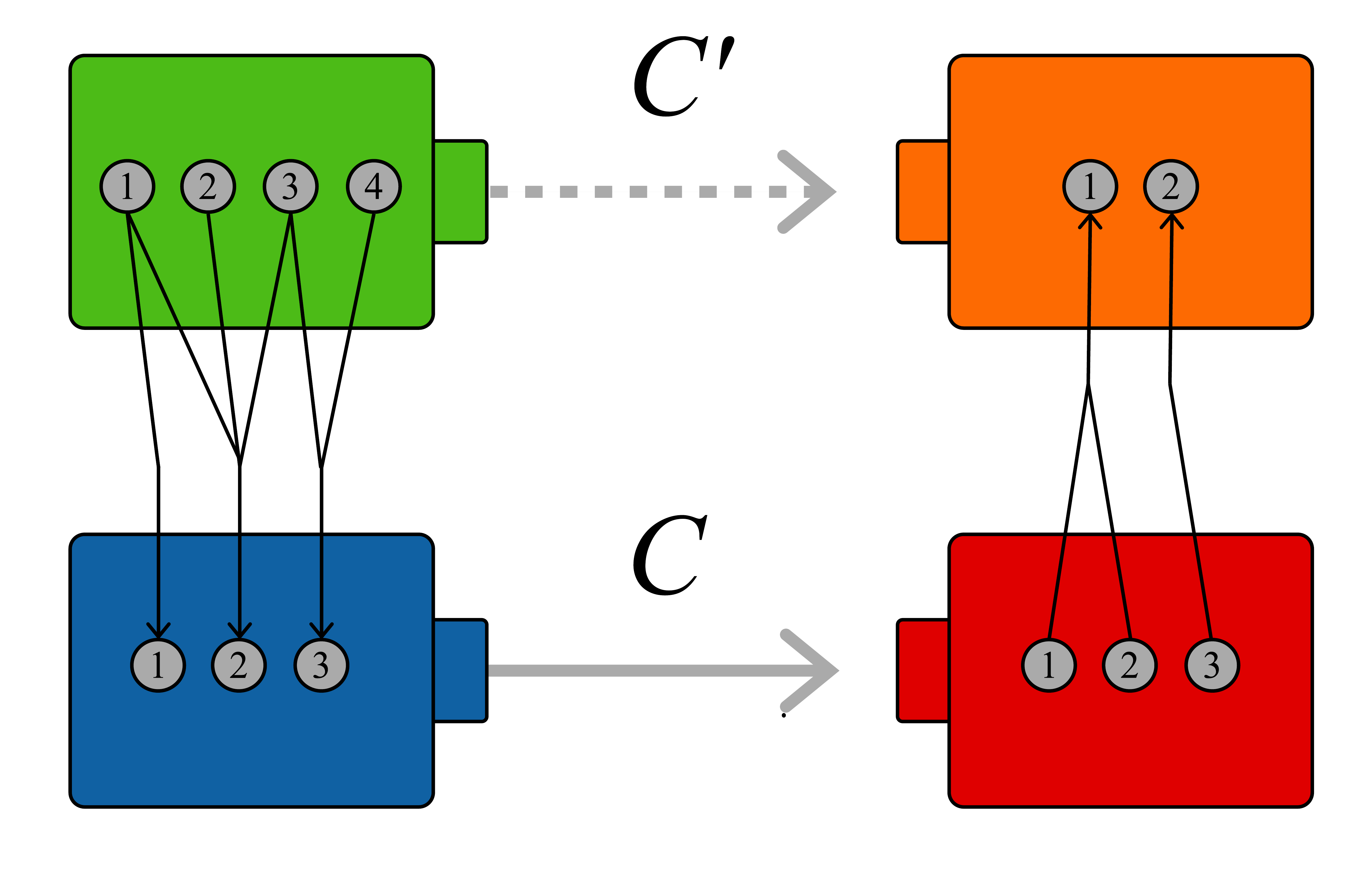}
\caption{\label{fig:ultraweak}If a matrix $C$ has an implementation with a physical system $\mathcal{S}$ and if a matrix $C'$ is ultraweakly majorized by $C$, then also $C'$ has an implementation with the system $\mathcal{S}$ as it can be obtained from the setup of $C$ by mixing and relabeling of states and measurement outcomes.}
\end{figure}

\begin{example}
By permuting the rows (or columns) of the antidiagonal matrix $C^{\mathrm{opt}}_{n,n-1}$ one can obtain the identity matrix $\id_n$.
Thereby, $C^{\mathrm{opt}}_{n,n-1} \simeq \id_n$.
This is consistent with the fact that if we have $n$ boxes and we know $n-1$ boxes where the candy is not going to be, then we actually know where the candy will be.
\end{example}

The following three results give conditions when two optimal communication matrices are in the ultraweak majorization relation.  

\begin{proposition}\label{prop:rank}
A necessary condition for $C^{\mathrm{opt}}_{n,t}\uleq C^{\mathrm{opt}}_{n',t'}$ is that $n\leq n'$.
\end{proposition}

\begin{proof}
It is enough to show that $\rank{C^{\mathrm{opt}}_{n,t}}=n$ as the matrix rank cannot increase in the matrix multiplication. 
Let $\mathcal{V}\subseteq\R^n$ be the space spanned by the rows of $C^{\mathrm{opt}}_{n,t}$ and denote by $\mathcal{V}^\bot$ the orthogonal complement of $\mathcal{V}$.
We first make an assumption that there exists $u\in\mathcal{V}^\bot$ such that $u_i \neq u_j$ for some indices $i < j$. 
We can always find a row $v$ in $C^{\mathrm{opt}}_{n,t}$ for which $v_i = 0$ and $v_j > 0$. 
Further, $C^{\mathrm{opt}}_{n,t}$ also contains a row $w$ for which $w_i = v_j$, $w_j=v_i$, and the other components are the same. 
We then obtain
$$
0 = u \cdot v - u \cdot w = u \cdot (v-w) = u_i(v_i-w_i) + u_j(v_j-w_j) = v_j(u_j-u_i) \neq 0 \, ,
$$
which is a contradiction.
Therefore, contrary to our assumption every $u\in\mathcal{V}^\bot$ is of the form $u=s (1,\dots,1)$ for some $s\in \mathbb{R}$. Moreover, because all rows of $C^{\mathrm{opt}}_{n,t}$ only contain non-negative components, we must have $s=0$. 
We conclude that $\mathcal{V}^\bot=\{0\}$ and hence $\rank{C^{\mathrm{opt}}_{n,t}}=\dim \mathcal{V}=n$.
\end{proof}

Our earlier finding $C^{\mathrm{opt}}_{3,1} \nuleq C^{\mathrm{opt}}_{2,1}$ in Example \ref{ex:2131} shows that the condition stated in Prop. \ref{prop:rank} is not sufficient.
The following results generalizes the earlier observation  $C^{\mathrm{opt}}_{2,1} \uleq C^{\mathrm{opt}}_{4,2}$ in Example \ref{sec:uwm-examples} by provinding a sufficient condition for the ultraweak majorization.

\begin{proposition}
A sufficient condition for $C^{\mathrm{opt}}_{m,m-1}\uleq C^{\mathrm{opt}}_{n,t}$ is that $\lfloor \tfrac{n}{n-t} \rfloor\geq m$.
\end{proposition}

\begin{proof}
The first row of $C^{\mathrm{opt}}_{n,t}$ is the vector written in \eqref{eq:firstrow}. As every distinct permutation of this vector appears in $C^{\mathrm{opt}}_{m,t}$, we can find $k := \lfloor \tfrac{m}{m-t} \rfloor$ rows so that each of them has nonzero elements only in locations where the other rows have zeros. 
For instance, in the case of $n=10$, $t=7$, we choose the following rows:
\begin{align*}
\tfrac{1}{3} \left[\begin{array}{cccccccccc} 0 & 0 & 0 & 0 & 0 & 0 & 0 & 1 & 1 & 1  \end{array}\right] \, , \\
\tfrac{1}{3} \left[\begin{array}{cccccccccc} 0 & 0 & 0 & 0 & 1 & 1 & 1 & 0 & 0 & 0  \end{array}\right] \, , \\
\tfrac{1}{3} \left[\begin{array}{cccccccccc} 0 & 1 & 1 & 1 & 0 & 0 & 0 & 0 & 0 & 0  \end{array}\right] \, .
\end{align*}
It is easy to see that we can choose a matrix $L$ in such away that $LC^{\mathrm{opt}}_{n,t}$ has only the previously chosen $k$ rows but is otherwise unchanged. 
Then we choose a matrix $R$ such that similar columns are added together. 
Finally, by Prop. \ref{prop:equivalent} we can remove zero columns if there are any.
In this way, we have obtained $C^{\mathrm{opt}}_{k,k-1}$.
\end{proof}

\begin{theorem}\label{diagonal proof}
$C^{\mathrm{opt}}_{n,t-1} \uleq C^{\mathrm{opt}}_{n,t}\uleq C^{\mathrm{opt}}_{n+1,t+1}$.
\end{theorem}

\begin{proof}
We first prove $C^{\mathrm{opt}}_{n,t}\uleq C^{\mathrm{opt}}_{n+1,t+1}$.

Firstly, observe that $C^{\mathrm{opt}}_{n+1,t+1}$ is the matrix whose rows are obtained by permuting the row vector 
$$
\tfrac{1}{n-t}\left[\begin{array}{cccccccc}0 & 0 & \cdots & 0 & 1 & 1 & \cdots & 1
\end{array}\right]$$
where there are $t+1$ zeros and $n-t$ ones. 
Similarly, $C^{\mathrm{opt}}_{n,t}$ is obtained by permuting the similar vector as above except for the fact that the first row vector has $t$ zeros. 
As the order of the rows of an optimal communication matrix does not matter (see Prop. \ref{prop:equivalent}), we can write $C^{\mathrm{opt}}_{n+1,t+1}$ in such a way that it contains $C^{\mathrm{opt}}_{n,t}$ as a submatrix in the top right corner. 

We construct $L$ as the ${n-1 \choose t-1} \times {n \choose t}$ matrix with a ${n-1 \choose t-1} \times {n-1 \choose t-1}$ identity matrix as a submatrix on the left. Then $L C^{\mathrm{opt}}_{n+1,t+1}$ will just be the matrix which contains the first ${n-1 \choose t-1}$ rows of $C^{\mathrm{opt}}_{n+1,t+1}$. 
The first column of this matrix is all zeros.

For $R$ we can choose the $n \times (n-1)$ matrix with a $(n-1) \times (n-1)$ identity matrix as a submatrix in the bottom. As the first column vector of $L C^{\mathrm{opt}}_{n+1,t+1}$ was all zeros, the first row of $R$ can be freely chosen. We can, for instance, put a one as the first element. Now $L C^{\mathrm{opt}}_{n+1,t+1} R = C^{\mathrm{opt}}_{n,t}$ by the construction of $L$ and $R$.
We have thus shown that $C^{\mathrm{opt}}_{n,t}\uleq C^{\mathrm{opt}}_{n+1,t+1}$.

For the relation $C^{\mathrm{opt}}_{n,t-1} \uleq C^{\mathrm{opt}}_{n,t}$, we begin by showing that the rows of $C^{\mathrm{opt}}_{n,t-1}$ belong to the convex hull of the rows of $C^{\mathrm{opt}}_{n,t}$. The matrix $C^{\mathrm{opt}}_{n,t}$ consists of all rows obtained by permuting the vector 
$$
\tfrac{1}{n-t}\left[\begin{array}{cccccccc}0 & 0 & \cdots & 0 & 1 & 1 & \cdots & 1
\end{array}\right] \, ,
$$
where there are $t$ zeros and $n-t$ ones. Likewise, the matrix $C^{\mathrm{opt}}_{n,t-1}$ is obtained by all permutations of the vector 
$$
\tfrac{1}{n-t+1}\left[\begin{array}{cccccccc}0 & 0 & \cdots & 0 & 1 & 1 & \cdots & 1
\end{array}\right] \, , 
$$
where there are $t-1$ zeros and $n-t+1$ ones. In particular, we can form any row of  $C^{\mathrm{opt}}_{n,t-1}$ by summing up those rows of $C^{\mathrm{opt}}_{n,t}$ that have zeros in the same places as the row in question, all multiplied with the factor $1/(n-t+1)$. 
There are always $n-t+1$ of these rows. Hence, the rows of $C^{\mathrm{opt}}_{n,t-1}$ belong to the convex hull of the rows of $C^{\mathrm{opt}}_{n,t}$.
It follows from \cite[Prop. 3.3.]{PeMaSi05} that there exists a row-stochastic matrix $L$ such that $LC^{\mathrm{opt}}_{n,t}=C^{\mathrm{opt}}_{n,t-1}$.
Therefore, $C^{\mathrm{opt}}_{n,t-1} \uleq C^{\mathrm{opt}}_{n,t}$.

For completeness, we construct a matrix $L$ that fulfils $LC^{\mathrm{opt}}_{n,t} = C^{\mathrm{opt}}_{n,t-1}$. 
We denote by $row_i(C)$ the $i$th row of a matrix $C$.
Let us consider the following equation: \begin{align*}
row_i(C^{\mathrm{opt}}_{n,t-1}) = \sum_{k=1}^{{n \choose t}} l_{ik} \ row_k(C^{\mathrm{opt}}_{n,t}), \quad i = 1,\dots, {n \choose t - 1},
\end{align*}where \begin{equation*}
l_{ik} = \left \{ \begin{aligned}
& \tfrac{1}{n-t+1}, && \text{if } row_i(C^{\mathrm{opt}}_{n,t-1})\cdot row_k(C^{\mathrm{opt}}_{n,t}) = \tfrac{1}{n-t+1} \\
& 0, && \text{otherwise} 
\end{aligned} \right.
\end{equation*}Notice that $\frac{1}{n-t+1}$ is the maximal value that the inner product $row_i(C^{\mathrm{opt}}_{n,t-1})\cdot row_k(C^{\mathrm{opt}}_{n,t})$ can have, and in this case the rows $row_i(C^{\mathrm{opt}}_{n,t-1})$ and $row_k(C^{\mathrm{opt}}_{n,t})$  share the maximal number of zeros. Because there are $n-t+1$ of these rows, the matrix $L$ is row-stochastic as constructed and fulfils $LC^{\mathrm{opt}}_{n,t} = C^{\mathrm{opt}}_{n,t-1}$.
\end{proof}

As discussed in Sec. \ref{sec:ultraweak}, the main question related to the present investigation is to characterize all pairs $(n,t)$ such that the optimal communication matrix $C^{\mathrm{opt}}_{n,t}$ can be implemented with a physical system $\mathcal{S}$.
We have answered this question for the qubit and rebit, while we managed to provide some partial answers for the qudit.
The previous results imply that some tables describing implementable communication tests are impossible.
Two such examples are presented in Fig. \ref{fig:comm_tab_imp}.
The full characterization of the ultraweak preordering of the optimal communication matrices $C^{\mathrm{opt}}_{n,t}$ remains an open question.

\begin{figure}
\centering
\begin{subfigure}{.5\textwidth}
\centering
\scalebox{1.2}{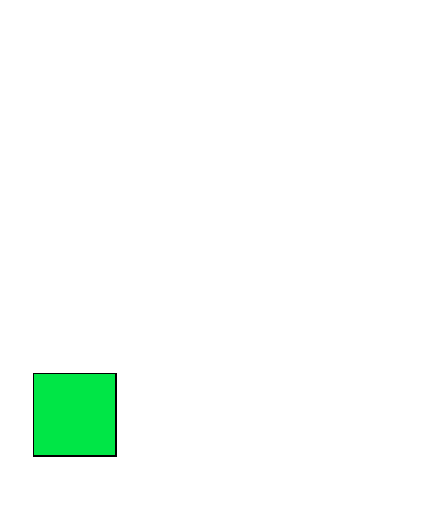}
\end{subfigure}%
\begin{subfigure}{.5\textwidth}
\centering
\scalebox{1.2}{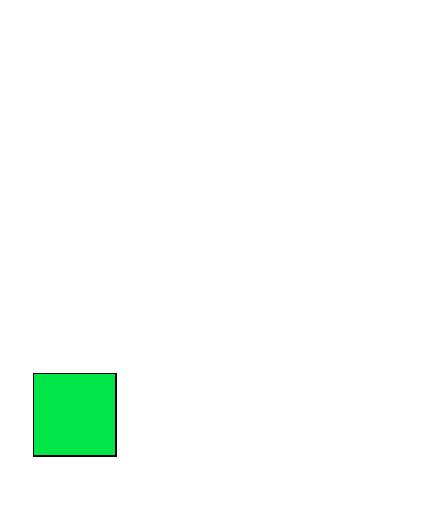}
\end{subfigure}
\caption{Two communication tables that by Theorem \ref{diagonal proof} cannot be related to any physical system.}
\label{fig:comm_tab_imp}
\end{figure}

%%%%%%%%%%%%%%%%%%%%%%%%%%%%%%%%%%
\section{Acknowledgement}
%%%%%%%%%%%%%%%%%%%%%%%%%%%%%%%%%%

This work was performed as part of the Academy of Finland Centre of Excellence program, Project 312058.
Financial support from the Academy of Finland, Project 287750, is also acknowledged.

%%%%%%%%%%%%%%%%%%%%
%%%%%%%%%%%%%%%%%%%%

%%%%%%%%%%
%%%%%%%%%%
\end{document}

%% file: communication_table_qubit.pdf_tex
%% Creator: Inkscape inkscape 0.92.2, www.inkscape.org
%% Accompanies image file 'communication_table_qubit.pdf' (pdf, eps, ps)
\begingroup%
  \makeatletter%
  \providecommand\color[2][]{%
    \errmessage{(Inkscape) Color is used for the text in Inkscape, but the package 'color.sty' is not loaded}%
    \renewcommand\color[2][]{}%
  }%
  \providecommand\transparent[1]{%
    \errmessage{(Inkscape) Transparency is used (non-zero) for the text in Inkscape, but the package 'transparent.sty' is not loaded}%
    \renewcommand\transparent[1]{}%
  }%
  \providecommand\rotatebox[2]{#2}%
  \ifx\svgwidth\undefined%
    \setlength{\unitlength}{127.55905512bp}%
    \ifx\svgscale\undefined%
      \relax%
    \else%
      \setlength{\unitlength}{\unitlength * \real{\svgscale}}%
    \fi%
  \else%
    \setlength{\unitlength}{\svgwidth}%
  \fi%
  \global\let\svgwidth\undefined%
  \global\let\svgscale\undefined%
  \makeatother%
  \begin{picture}(1,1.15555556)%
    \put(0.17344907,0.20872735){\color[rgb]{0,0,0}\makebox(0,0)[lb]{\smash{}}}%
    \put(0,0){\includegraphics[width=\unitlength,page=1]{communication_table_qubit.pdf}}%
    \put(0.1088849,0.18845825){\color[rgb]{0,0,0}\makebox(0,0)[lb]{\smash{2-1}}}%
    \put(0,0){\includegraphics[width=\unitlength,page=2]{communication_table_qubit.pdf}}%
    \put(0.84831696,0.03461711){\color[rgb]{0,0,0}\makebox(0,0)[lb]{\smash{$n$}}}%
    \put(-0.02286751,0.90766127){\color[rgb]{0,0,0}\makebox(0,0)[lb]{\smash{$t$}}}%
    \put(0,0){\includegraphics[width=\unitlength,page=3]{communication_table_qubit.pdf}}%
    \put(0.29479404,0.18845825){\color[rgb]{0,0,0}\makebox(0,0)[lb]{\smash{3-1}}}%
    \put(0,0){\includegraphics[width=\unitlength,page=4]{communication_table_qubit.pdf}}%
    \put(0.48070314,0.18845825){\color[rgb]{0,0,0}\makebox(0,0)[lb]{\smash{4-1}}}%
    \put(0,0){\includegraphics[width=\unitlength,page=5]{communication_table_qubit.pdf}}%
    \put(0.66661229,0.18845825){\color[rgb]{0,0,0}\makebox(0,0)[lb]{\smash{5-1}}}%
    \put(0,0){\includegraphics[width=\unitlength,page=6]{communication_table_qubit.pdf}}%
    \put(0.29479404,0.37436727){\color[rgb]{0,0,0}\makebox(0,0)[lb]{\smash{3-2}}}%
    \put(0,0){\includegraphics[width=\unitlength,page=7]{communication_table_qubit.pdf}}%
    \put(0.48070314,0.37436727){\color[rgb]{0,0,0}\makebox(0,0)[lb]{\smash{4-2}}}%
    \put(0,0){\includegraphics[width=\unitlength,page=8]{communication_table_qubit.pdf}}%
    \put(0.66661229,0.37436727){\color[rgb]{0,0,0}\makebox(0,0)[lb]{\smash{5-2}}}%
    \put(0,0){\includegraphics[width=\unitlength,page=9]{communication_table_qubit.pdf}}%
    \put(0.48070314,0.56027629){\color[rgb]{0,0,0}\makebox(0,0)[lb]{\smash{4-3}}}%
    \put(0,0){\includegraphics[width=\unitlength,page=10]{communication_table_qubit.pdf}}%
    \put(0.66661229,0.56027629){\color[rgb]{0,0,0}\makebox(0,0)[lb]{\smash{5-3}}}%
    \put(0,0){\includegraphics[width=\unitlength,page=11]{communication_table_qubit.pdf}}%
    \put(0.6666122,0.7461853){\color[rgb]{0,0,0}\makebox(0,0)[lb]{\smash{5-4}}}%
    \put(0,0){\includegraphics[width=\unitlength,page=12]{communication_table_qubit.pdf}}%
  \end{picture}%
\endgroup%

%% file: communication_table_rebit.pdf_tex
%% Creator: Inkscape inkscape 0.92.2, www.inkscape.org
%% Accompanies image file 'communication_table_rebit.pdf' (pdf, eps, ps)
\begingroup%
  \makeatletter%
  \providecommand\color[2][]{%
    \errmessage{(Inkscape) Color is used for the text in Inkscape, but the package 'color.sty' is not loaded}%
    \renewcommand\color[2][]{}%
  }%
  \providecommand\transparent[1]{%
    \errmessage{(Inkscape) Transparency is used (non-zero) for the text in Inkscape, but the package 'transparent.sty' is not loaded}%
    \renewcommand\transparent[1]{}%
  }%
  \providecommand\rotatebox[2]{#2}%
  \ifx\svgwidth\undefined%
    \setlength{\unitlength}{127.55905512bp}%
    \ifx\svgscale\undefined%
      \relax%
    \else%
      \setlength{\unitlength}{\unitlength * \real{\svgscale}}%
    \fi%
  \else%
    \setlength{\unitlength}{\svgwidth}%
  \fi%
  \global\let\svgwidth\undefined%
  \global\let\svgscale\undefined%
  \makeatother%
  \begin{picture}(1,1.15555556)%
    \put(0.17344907,0.20872735){\color[rgb]{0,0,0}\makebox(0,0)[lb]{\smash{}}}%
    \put(0,0){\includegraphics[width=\unitlength,page=1]{communication_table_rebit.pdf}}%
    \put(0.1088849,0.18845825){\color[rgb]{0,0,0}\makebox(0,0)[lb]{\smash{2-1}}}%
    \put(0,0){\includegraphics[width=\unitlength,page=2]{communication_table_rebit.pdf}}%
    \put(0.84831696,0.03461711){\color[rgb]{0,0,0}\makebox(0,0)[lb]{\smash{$n$}}}%
    \put(-0.02286751,0.90766127){\color[rgb]{0,0,0}\makebox(0,0)[lb]{\smash{$t$}}}%
    \put(0,0){\includegraphics[width=\unitlength,page=3]{communication_table_rebit.pdf}}%
    \put(0.29479404,0.18845825){\color[rgb]{0,0,0}\makebox(0,0)[lb]{\smash{3-1}}}%
    \put(0,0){\includegraphics[width=\unitlength,page=4]{communication_table_rebit.pdf}}%
    \put(0.48070314,0.18845825){\color[rgb]{0,0,0}\makebox(0,0)[lb]{\smash{4-1}}}%
    \put(0,0){\includegraphics[width=\unitlength,page=5]{communication_table_rebit.pdf}}%
    \put(0.66661229,0.18845825){\color[rgb]{0,0,0}\makebox(0,0)[lb]{\smash{5-1}}}%
    \put(0,0){\includegraphics[width=\unitlength,page=6]{communication_table_rebit.pdf}}%
    \put(0.29479404,0.37436727){\color[rgb]{0,0,0}\makebox(0,0)[lb]{\smash{3-2}}}%
    \put(0,0){\includegraphics[width=\unitlength,page=7]{communication_table_rebit.pdf}}%
    \put(0.48070314,0.37436727){\color[rgb]{0,0,0}\makebox(0,0)[lb]{\smash{4-2}}}%
    \put(0,0){\includegraphics[width=\unitlength,page=8]{communication_table_rebit.pdf}}%
    \put(0.66661229,0.37436727){\color[rgb]{0,0,0}\makebox(0,0)[lb]{\smash{5-2}}}%
    \put(0,0){\includegraphics[width=\unitlength,page=9]{communication_table_rebit.pdf}}%
    \put(0.48070314,0.56027629){\color[rgb]{0,0,0}\makebox(0,0)[lb]{\smash{4-3}}}%
    \put(0,0){\includegraphics[width=\unitlength,page=10]{communication_table_rebit.pdf}}%
    \put(0.66661229,0.56027629){\color[rgb]{0,0,0}\makebox(0,0)[lb]{\smash{5-3}}}%
    \put(0,0){\includegraphics[width=\unitlength,page=11]{communication_table_rebit.pdf}}%
    \put(0.6666122,0.7461853){\color[rgb]{0,0,0}\makebox(0,0)[lb]{\smash{5-4}}}%
    \put(0,0){\includegraphics[width=\unitlength,page=12]{communication_table_rebit.pdf}}%
  \end{picture}%
\endgroup%

%% file: impossible_communication_1.pdf_tex
%% Creator: Inkscape inkscape 0.92.2, www.inkscape.org
%% Accompanies image file 'impossible_communication_1.pdf' (pdf, eps, ps)
\begingroup%
  \makeatletter%
  \providecommand\color[2][]{%
    \errmessage{(Inkscape) Color is used for the text in Inkscape, but the package 'color.sty' is not loaded}%
    \renewcommand\color[2][]{}%
  }%
  \providecommand\transparent[1]{%
    \errmessage{(Inkscape) Transparency is used (non-zero) for the text in Inkscape, but the package 'transparent.sty' is not loaded}%
    \renewcommand\transparent[1]{}%
  }%
  \providecommand\rotatebox[2]{#2}%
  \ifx\svgwidth\undefined%
    \setlength{\unitlength}{127.55905512bp}%
    \ifx\svgscale\undefined%
      \relax%
    \else%
      \setlength{\unitlength}{\unitlength * \real{\svgscale}}%
    \fi%
  \else%
    \setlength{\unitlength}{\svgwidth}%
  \fi%
  \global\let\svgwidth\undefined%
  \global\let\svgscale\undefined%
  \makeatother%
  \begin{picture}(1,1.15555556)%
    \put(0.17344907,0.20872735){\color[rgb]{0,0,0}\makebox(0,0)[lb]{\smash{}}}%
    \put(0,0){\includegraphics[width=\unitlength,page=1]{impossible_communication_1.pdf}}%
    \put(0.1088849,0.18845825){\color[rgb]{0,0,0}\makebox(0,0)[lb]{\smash{2-1}}}%
    \put(0,0){\includegraphics[width=\unitlength,page=2]{impossible_communication_1.pdf}}%
    \put(0.84831696,0.03461711){\color[rgb]{0,0,0}\makebox(0,0)[lb]{\smash{$n$}}}%
    \put(-0.02286751,0.90766127){\color[rgb]{0,0,0}\makebox(0,0)[lb]{\smash{$t$}}}%
    \put(0,0){\includegraphics[width=\unitlength,page=3]{impossible_communication_1.pdf}}%
    \put(0.29479404,0.18845825){\color[rgb]{0,0,0}\makebox(0,0)[lb]{\smash{3-1}}}%
    \put(0,0){\includegraphics[width=\unitlength,page=4]{impossible_communication_1.pdf}}%
    \put(0.48070314,0.18845825){\color[rgb]{0,0,0}\makebox(0,0)[lb]{\smash{4-1}}}%
    \put(0,0){\includegraphics[width=\unitlength,page=5]{impossible_communication_1.pdf}}%
    \put(0.66661229,0.18845825){\color[rgb]{0,0,0}\makebox(0,0)[lb]{\smash{5-1}}}%
    \put(0,0){\includegraphics[width=\unitlength,page=6]{impossible_communication_1.pdf}}%
    \put(0.29479404,0.37436727){\color[rgb]{0,0,0}\makebox(0,0)[lb]{\smash{3-2}}}%
    \put(0,0){\includegraphics[width=\unitlength,page=7]{impossible_communication_1.pdf}}%
    \put(0.48070314,0.37436727){\color[rgb]{0,0,0}\makebox(0,0)[lb]{\smash{4-2}}}%
    \put(0,0){\includegraphics[width=\unitlength,page=8]{impossible_communication_1.pdf}}%
    \put(0.66661229,0.37436727){\color[rgb]{0,0,0}\makebox(0,0)[lb]{\smash{5-2}}}%
    \put(0,0){\includegraphics[width=\unitlength,page=9]{impossible_communication_1.pdf}}%
    \put(0.48070314,0.56027629){\color[rgb]{0,0,0}\makebox(0,0)[lb]{\smash{4-3}}}%
    \put(0,0){\includegraphics[width=\unitlength,page=10]{impossible_communication_1.pdf}}%
    \put(0.66661229,0.56027629){\color[rgb]{0,0,0}\makebox(0,0)[lb]{\smash{5-3}}}%
    \put(0,0){\includegraphics[width=\unitlength,page=11]{impossible_communication_1.pdf}}%
    \put(0.6666122,0.7461853){\color[rgb]{0,0,0}\makebox(0,0)[lb]{\smash{5-4}}}%
    \put(0,0){\includegraphics[width=\unitlength,page=12]{impossible_communication_1.pdf}}%
  \end{picture}%
\endgroup%

%% file: impossible_communication_2.pdf_tex
%% Creator: Inkscape inkscape 0.92.2, www.inkscape.org
%% Accompanies image file 'impossible_communication_2.pdf' (pdf, eps, ps)
\begingroup%
  \makeatletter%
  \providecommand\color[2][]{%
    \errmessage{(Inkscape) Color is used for the text in Inkscape, but the package 'color.sty' is not loaded}%
    \renewcommand\color[2][]{}%
  }%
  \providecommand\transparent[1]{%
    \errmessage{(Inkscape) Transparency is used (non-zero) for the text in Inkscape, but the package 'transparent.sty' is not loaded}%
    \renewcommand\transparent[1]{}%
  }%
  \providecommand\rotatebox[2]{#2}%
  \ifx\svgwidth\undefined%
    \setlength{\unitlength}{127.55905512bp}%
    \ifx\svgscale\undefined%
      \relax%
    \else%
      \setlength{\unitlength}{\unitlength * \real{\svgscale}}%
    \fi%
  \else%
    \setlength{\unitlength}{\svgwidth}%
  \fi%
  \global\let\svgwidth\undefined%
  \global\let\svgscale\undefined%
  \makeatother%
  \begin{picture}(1,1.15555556)%
    \put(0.17344907,0.20872735){\color[rgb]{0,0,0}\makebox(0,0)[lb]{\smash{}}}%
    \put(0,0){\includegraphics[width=\unitlength,page=1]{impossible_communication_2.pdf}}%
    \put(0.1088849,0.18845825){\color[rgb]{0,0,0}\makebox(0,0)[lb]{\smash{2-1}}}%
    \put(0,0){\includegraphics[width=\unitlength,page=2]{impossible_communication_2.pdf}}%
    \put(0.84831696,0.03461711){\color[rgb]{0,0,0}\makebox(0,0)[lb]{\smash{$n$}}}%
    \put(-0.02286751,0.90766127){\color[rgb]{0,0,0}\makebox(0,0)[lb]{\smash{$t$}}}%
    \put(0,0){\includegraphics[width=\unitlength,page=3]{impossible_communication_2.pdf}}%
    \put(0.29479404,0.18845825){\color[rgb]{0,0,0}\makebox(0,0)[lb]{\smash{3-1}}}%
    \put(0,0){\includegraphics[width=\unitlength,page=4]{impossible_communication_2.pdf}}%
    \put(0.48070314,0.18845825){\color[rgb]{0,0,0}\makebox(0,0)[lb]{\smash{4-1}}}%
    \put(0,0){\includegraphics[width=\unitlength,page=5]{impossible_communication_2.pdf}}%
    \put(0.66661229,0.18845825){\color[rgb]{0,0,0}\makebox(0,0)[lb]{\smash{5-1}}}%
    \put(0,0){\includegraphics[width=\unitlength,page=6]{impossible_communication_2.pdf}}%
    \put(0.29479404,0.37436727){\color[rgb]{0,0,0}\makebox(0,0)[lb]{\smash{3-2}}}%
    \put(0,0){\includegraphics[width=\unitlength,page=7]{impossible_communication_2.pdf}}%
    \put(0.48070314,0.37436727){\color[rgb]{0,0,0}\makebox(0,0)[lb]{\smash{4-2}}}%
    \put(0,0){\includegraphics[width=\unitlength,page=8]{impossible_communication_2.pdf}}%
    \put(0.66661229,0.37436727){\color[rgb]{0,0,0}\makebox(0,0)[lb]{\smash{5-2}}}%
    \put(0,0){\includegraphics[width=\unitlength,page=9]{impossible_communication_2.pdf}}%
    \put(0.48070314,0.56027629){\color[rgb]{0,0,0}\makebox(0,0)[lb]{\smash{4-3}}}%
    \put(0,0){\includegraphics[width=\unitlength,page=10]{impossible_communication_2.pdf}}%
    \put(0.66661229,0.56027629){\color[rgb]{0,0,0}\makebox(0,0)[lb]{\smash{5-3}}}%
    \put(0,0){\includegraphics[width=\unitlength,page=11]{impossible_communication_2.pdf}}%
    \put(0.6666122,0.7461853){\color[rgb]{0,0,0}\makebox(0,0)[lb]{\smash{5-4}}}%
    \put(0,0){\includegraphics[width=\unitlength,page=12]{impossible_communication_2.pdf}}%
  \end{picture}%
\endgroup%

%% file: arxiv-partial_ignorance_v2.bbl
\begin{thebibliography}{10}

\bibitem{QPSI10}
B.~Schumacher and M.~Westmoreland.
\newblock {\em {Q}uantum {P}rocesses, {S}ystems, and {I}nformation}.
\newblock Cambridge University Press, 2010.

\bibitem{Hardy03}
L.~Hardy.
\newblock Probability theories in general and quantum theory in particular.
\newblock {\em Stud. Hist. Philos. Modern Phys.}, 34:381--393, 2003.

\bibitem{Zauner99}
G.~Zauner.
\newblock Quantendesigns. {G}rundz\"uge einer nichtkommutativen designtheorie.
\newblock PhD Thesis. University of Vienna., 1999.
\newblock Quantum designs: foundations of a noncommutative design theory. {\em Int.
  J. Quantum Inf.} 9:445--508, 2011 (published in english translation).

\bibitem{ReBlScCa04}
J.M. Renes, R.~Blume-Kohout, A.J. Scott, and C.M. Caves.
\newblock Symmetric informationally complete quantum measurements.
\newblock {\em J. Math. Phys.}, 45:2171--2180, 2004.

\bibitem{Dahl99}
G.~Dahl.
\newblock Matrix majorization.
\newblock {\em Linear Algebra Appl.}, 288:53--73, 1999.

\bibitem{PeMaSi05}
F.D.~Mart{\'\i}nez Per{\'\i}a, P.G. Massey, and L.E. Silvestre.
\newblock Weak matrix majorization.
\newblock {\em Linear Algebra Appl.}, 403:343--368, 2005.

\bibitem{Leifer14}
M.~Leifer.
\newblock Is the {Q}uantum {S}tate {R}eal? {A}n {E}xtended {R}eview of
  $\psi$-ontology {T}heorems.
\newblock {\em Quanta}, 3:67--155, 2014.

\bibitem{HeKe18}
T.~Heinosaari and O.~Kerppo.
\newblock Antidistinguishability of pure quantum states.
\newblock {\em J. Phys. A: Math. Theor.}, 51:365303, 2018.

\bibitem{BaJaOpPe14}
S.~Bandyopadhyay, R.~Jain, J.~Oppenheim, and C.~Perry.
\newblock Conclusive exclusion of quantum states.
\newblock {\em Phys. Rev. A}, 89:022336, 2014.

\bibitem{CaFuSc02}
C.M. Caves, C.A. Fuchs, and R.~Schack.
\newblock Conditions for compatibility of quantum-state assignments.
\newblock {\em Phys. Rev. A}, 66:062111, 2002.

\bibitem{Molina19}
A.~Molina.
\newblock {POVM}s are equivalent to projections for perfect state exclusion of
  three pure states in three dimensions.
\newblock {\em {Quantum}}, 3:117, 2019.

\bibitem{ApChFlWa18}
M.~Appleby, T.-Y. Chien, S.~Flammia, and S.~Waldron.
\newblock Constructing exact symmetric informationally complete measurements
  from numerical solutions.
\newblock {\em J. Phys. A: Math. Theor.}, 51:165302, 2018.

\bibitem{KiNuIm10}
G.~Kimura, K.~Nuida, and H.~Imai.
\newblock Distinguishability measures and entropies for general probabilistic
  theories.
\newblock {\em Rep. Math. Phys.}, 66:175--206, 2010.

\bibitem{MaMu11}
L.~Masanes and M.~M\"uller.
\newblock A derivation of quantum theory from physical requirements.
\newblock {\em New J. Phys.}, 13:063001, 2011.

\end{thebibliography}
